\documentclass[twocolumn,journal,10pt]{IEEEtran}
\usepackage{amsmath,amssymb,amsfonts}
\usepackage{amsthm}
\usepackage{mathtools}
\usepackage[table]{xcolor} 
\usepackage{booktabs} 
\usepackage{eqparbox}
\usepackage{float}
\usepackage{algorithm}
\usepackage{longtable}
\usepackage{graphicx}
\usepackage{epstopdf}
\usepackage[noend]{algpseudocode}
\usepackage{algorithmicx}
\usepackage{subfig}
\usepackage{xcolor}
\usepackage[utf8]{inputenc}
\usepackage{caption}
\usepackage{textcomp}
\usepackage{lettrine}
\usepackage{algpseudocode}
\usepackage{xspace}
\usepackage{multirow}
\usepackage{siunitx}
\usepackage{acronym}
\usepackage{array}   
\usepackage{booktabs} 
\usepackage{caption} 
\usepackage{subcaption} 
\acrodef{NCR}{network-controlled repeater}
\acrodef{HO}[HO]{Handover}
\acrodef{UL}{uplink}
\acrodef{DL}{downlink}
\acrodef{3GPP}{3rd generation partnership project}
\acrodef{RIS}[RIS]{reconfigurable intelligent surface}
\acrodef{5G}{5th generation}
\acrodef{6G}[6G]{6th generation}
\acrodef{BS}[BS]{base station}
\acrodef{PL}[PL]{path-loss}
\acrodef{HOD}{HO decision} 
\acrodef{MIMO}{multiple input multiple output}
\acrodef{SNR}{signal-to-noise ratio}
\acrodef{SINR}[SINR]{signal-to-interference-plus-noise ratio}
\acrodef{SDP}{semidefinite programming}
\acrodef{CDF}{cumulative distribution function} 
\acrodef{CSI}{channel state information}
\acrodef{gNB}[BS]{base station}
\acrodef{AWGN}{additive white Gaussian noise}
\acrodef{CCI}{co-channel interference}
\acrodef{ULA}{uniform linear array}
\acrodef{RCS}{radar cross-section}
\acrodef{LoS}{line of sight} 
\acrodef{UE}[UE]{user equipment}
\acrodef{OFDM}{orthogonal frequency division multiplexing}
\acrodef{DoF}{degrees of freedom}
\acrodef{MI}{mutual interference}
\acrodef{Uu}{access link interface} 
\acrodef{TTT}{time-to-trigger}
\acrodef{HOM}{handover margin} 
\acrodef{RSRP}{received signal reference power}
\acrodef{RSRQ}[RSRQ]{received signal reference quality}
\acrodef{RLF}{radio link failure}
\acrodef{KPI}{key performance indicator}
\acrodef{HOC}{handover control}
\acrodef{HCP}{handover control parameter}
\acrodef{HOP}{handover probability} 
\acrodef{HOPP}{handover ping-pong} 
\acrodef{HCPs}{handover control parameters}
\acrodef{IoT}{internet of things}
\acrodef{V2X}{vehicle-to-everything}
\acrodef{UAV}{unmanned aerial vehicle}
\acrodef{AI}{artificial intelligence}
\acrodef{QoS}{quality of service}
\acrodef{THz}[THz]{terahertz}
\acrodef{NF}[NF]{near-field}
\acrodef{FF}[FF]{far-field}
\acrodef{UPW}[UPW]{uniform planewave}
\acrodef{NUSW}[NUSW]{non-uniform spherical wave}
\acrodef{SE}[SE]{ spectral efficiency}
\definecolor{lighterviolet}{RGB}{220,200,240}
\def\etc{etc.\@\xspace}
\usepackage[backend=biber,style=ieee]{biblatex} 
\usepackage{hyperref} 
\hypersetup{
    colorlinks=true, 
    linkcolor=blue,  
    citecolor=blue,  
    urlcolor=cyan    
}
\makeatletter
\renewcommand\@cite[2]{\textcolor{blue}{[}\textcolor{blue}{#1}\textcolor{blue}{]}} 

\newtheorem{proposition}{Proposition}  

\begin{document}

\title{Efficient handover based on Near-field and Far-field RIS for seamless connectivity}

\author{Atiquzzaman Mondal, Waheeb Tashan,~\IEEEmembership{Member,~IEEE}, Ayat Al-Olaimat, Hüseyin Arslan~\IEEEmembership{Fellow,~IEEE}

\thanks{A. Mondal, A. Olaimat and H. Arslan are with Department of Electrical and Electronics Engineering, Istanbul Medipol University, 34810 Istanbul, Türkiye.\\  W. Tashan is with Department of Electronics and Communication Engineering. Kocaeli University, 41001, Kocaeli, Türkiye}

}

{}

\maketitle
\begin{abstract}
Reconfigurable Intelligent Surfaces (RIS) is becoming a transformative technology for the upcoming 6G communication networks, providing a way for smartly maneuvering the electromagnetic waves to enhance coverage and connectivity. This paper presents an efficient handover (HO) management scheme leveraging RIS in the Fresnel region i.e., in both the near-field (NF) and far-field (FF) regions to reduce signaling overhead and optimize mobility management. For this, we analyzed the signal strength variations in the considered RIS-aided networks, considering the radiative NF and FF regions, and derive the probability density function (PDF) of the RIS-UE distance in the NF region to quantify RIS reflection gains along the user equipment (UE) trajectory. We propose a new HO algorithm incorporating several HO categories like hard handover (HHO), soft handover (SHO), RIS-aided cell breathing (RIS-CB), and RIS-aided ping-pong avoidance (RIS-PP) strategies. The proposed algorithm uses bit error rate (BER) as a key parameter to predict the minimization of unnecessary HOs by using RIS-aided pathways to retain connectivity with the serving base station (BS), which minimizes the requirement for frequent target BS searching and ultimately optimizes the HO. By restricting measurement reports and HO requests, the suggested method improves spectrum efficiency (SE) and energy efficiency (EE), especially in crowded cellular networks. Numerical results highlight significant reductions in HO rates and signaling load, ensuring seamless connectivity and improved quality of service (QoS) in 6G systems.

\end{abstract}


\section{Introduction}
\lettrine{T}{he} highly anticipated \ac{6G} is expected to have higher throughput, enhanced coverage, improved reliability, and very low end-to-end latency \etc as compared to its 3GPP predecessors~\cite{bassoli2021we}. To fulfill these ever-increasing demands for the next-generation wireless communication networks, new technologies are rapidly coming into the picture. In the evolution of B5G/6G communication networks, one such technology, \textit{i.e.}, \ac{RIS} is emerging as a promising one, that is strategically designed to intelligently control the communication signals by reconfiguring its reflecting elements, promoting the concept of smart radio environments, eventually aims to deliver nearly 100\% global coverage \cite{renzo2019smart, wu2019towards, yu2021smart}.
By manipulating the electromagnetic waves, RIS can enable communication signals to overcome obstacles, thereby extending network coverage and enhancing connectivity \cite{liao2023optimized, 9966212}.

Whether deployed in indoor or outdoor landscapes, RISs contribute to shaping smart radio environments that adapt dynamically to varying connectivity needs. Focusing on outdoor applications, RISs serve as a powerful tool for managing connectivity links and optimizing HO mechanisms \cite{8926369, 10041749}. With RIS technology, the HO paradigm is redefined. Instead of transitioning between \acp{BS}, RIS enhances the existing connectivity path by dynamically redirecting signals through optimized reflection i.e., it can seamlessly shift the transmission path from a direct link to an RIS-enhanced link, improving overall performance while maintaining the same transmitter node. This approach not only reduces the traffic overhead associated with traditional \ac{HO} management but also enhances network security by preventing unnecessary transmitter switches \cite{10380573}.

In addition, implementing ultra-dense small \ac{BS}s cause a high amount of HO. Therefore, due to the limitations mentioned above, maintaining the quality connection during the HO process becomes highly significant. The primary objective of HO is to seamlessly transfer the \ac{UE} to a more favorable link with superior link quality when the current link experiences degradation, such as in \ac{FF} scenarios. Given the significant signaling overhead associated with HOs, particularly in FF regions, efficient and optimal HO mechanisms are crucial to maintaining seamless service continuity in this evolving wireless landscape.
HO triggering relies on various decision indicators to preserve the quality connection. Common metrics include signal strength (RSRP)~\cite{gannapathy2023smart}, quality measures (RSRQ)~\cite{sun2021multi}, and interference-aware metrics such as SINR \cite{tashan2024optimal}.

While RIS technology has significantly propelled 6G advancement by addressing diverse scenarios and KPIs, its practical realization often involves large-scale deployments with numerous reflecting elements and high-frequency operation. Consequently, the Rayleigh distance, defined as $d_{FF} = \frac{2D^{2}}{\lambda}$, expands proportionally to the carrier frequency $f_{c}$ and the dimensions of the RIS with horizontal and vertical length $L_{y}$ and $L_{z}$ respectively i.e., $D=\sqrt{L^2_{y}+L^2_{z}}$. Unlike traditional \ac{UPW} propagation in FF communication, NF communication leverages the \ac{NUSW} model, enabling the joint estimation of both the angle and distance of a target. This two-dimensional estimation introduces opportunities to enhance the functionalities and efficiency of 6G systems. Unlike conventional wireless systems, RIS-assisted systems significantly expand the NF region, creating a virtual NF as well as FF in addition to the fields of the serving \ac{BS}, allowing precise estimation of angle, distance information, and coverage for targets within this area, particularly within the context of HO procedures.

Though adaptive switching is beneficial in RIS assisted communication, its performance can be limited by the accessibility of the \ac{BS}-RIS-UE path or RIS placement for ensuring effective communication \cite{10471522, 8936989, 8839948}. Recent studies have concentrated on theoretical models that not only formulate expressions for HO rates but also examine the influence of blockages, as detailed in \cite{9297352, 9615850, 10078238}. In general, RIS-assisted HO offers a promising solution to reduce redundant HOs in wireless networks. However, designing an HO strategy that fully harnesses RIS’s potential in the NF to minimize the number of HOs or find an optimal HO remains a significant challenge. In dense cellular networks, one key challenge is the significant signaling overhead generated during the HO preparation phase. Consequently, proximity to a \ac{BS} does not guarantee a stronger signal, as the RIS enhances only the serving \ac{BS}'s signal, altering the traditional HO locations. In \cite{9003219}, distance-based HO models fail to capture the complexities of RIS-assisted HO. The HO rates and the RIS reconfigurations have been examined in \cite{gao2024intelligent} where the UE can remain connected to the serving \ac{BS} through RIS or initiate the HO to the target \ac{BS}. Jiao et al. \cite{jiao2021enabling} have proposed a RIS-assisted approach to minimize unnecessary HOs and maximize spectral efficiency using deep reinforcement learning. On top of that, a discrete-time model is presented in \cite{zhang2024discrete} to model the fluctuation of signal strength for accurate HO analysis of RIS-assisted networks. Wei et al. \cite{wei2023equivalent} have proposed an equivalent model of user-\ac{BS} distance to present the problem of the RIS cascaded channel to determine the optimal deployment factor for RIS distribution that reduces the HOP. A brief summary of existing research relevant to our study is provided in Table \ref{Comparison table}. 

Since the serving BS must identify and coordinate with multiple candidate target BSs simultaneously, the volume of HO requests increases and drives up energy consumption across both serving and target cells \cite{cui2022near, zhi2024performance, shen2023multi}. Furthermore, the UEs need to continuously relay measurement reports (e.g., signal quality metrics) to facilitate accurate HO decisions, which further exacerbates signaling load and diminishes \ac{SE} by consuming valuable uplink resources \cite{8643739}. Consequently, dense deployments face increased control-plane burden, reduced energy efficiency, and lower overall throughput, highlighting the need for more effective HO strategies that can alleviate these performance bottlenecks\cite{liu2022deep, peng2024deep, han2020channel, cui2022channel}. 
\par Although recent studies have demonstrated the potential of RIS to improve HO performance in mmWave and sub-THz networks, the majority of these works have focused on far-field scenarios, often assuming idealized signal propagation and limited mobility models \cite{wei2025handover,wei2025analysis}.
These approaches typically neglect the unique propagation characteristics of the NF region, particularly in dense deployments and higher frequencies, where the Rayleigh distance considerably increases due to significant \ac{RIS} apertures and operational wavelengths~\cite{wang2024wideband, zhou2024near}.
Furthermore, most current HO strategies are reactive, being triggered only after a deterioration in link quality. They do not incorporate predictive mechanisms that could leverage \acp{RIS} ability to dynamically reshape wavefronts and adaptively extend coverage~\cite{jiao2021enabling}. 
Bensalem et al. \cite{bensalem2024signaling} have presented a stochastic model to explore the rate of changes in RIS settings and HO rate to control the signaling rate. Moreover, the distance is considered as the HO decision algorithm. In addition, the authors in \cite{adnan2024performance} have presented a blockage prediction method using RIS-assisted mmWave, but did not address proactive HO management based on NF propagation characteristics.

To overcome the identified challenges, with the assistance of RIS placed at the edge of the quadratic near field, we provide an efficient HO scheme. To the best of our knowledge, none of the aforementioned studies RIS assisted HO in the NF region. Additionally, there is no existing work focused on leveraging the reflecting capabilities of RIS in this context to minimize unnecessary HOs or optimize HO decisions. The main contributions of this work are outlined below:
\begin{itemize}
    \item Taking into account both the radiative NF and FF regions, we theoretically analyze serving BS signal strength changes during HO in a typical (non-RIS-assisted) scenario.
    \item We measure the signal strength fluctuation after deploying a RIS to create a virtual NF and FF to extend the coverage area of the serving BS, thereby mitigating signaling overhead and improving HO efficiency.
    \item Upon establishing the RIS connection within the NF or FF region, the geometric relationship between the HO, UE and the RIS is continuously tracked. Leveraging the correlation between measurement intervals, the probability density function (PDF) of the RIS-UE distance is derived. This derived PDF enables the calculation of the bit error rate (BER), outage probability (OP) and capacity, which are expected to trigger the HO process.
    \item Reducing the number of HO requests by limiting the potential target \ac{BS}s to a minimum (preferably a single target BS) by performing HO probability analysis. The RIS-assisted NF and FF regions will be strategically positioned to align with the vicinity of the target BS, ensuring seamless mobility management.
\end{itemize}
The organization of this work is done as follows. The system framework for the considered RIS-assisted NF networks is introduced in Section II. Section III details the mathematical modeling. Numerical results demonstrating the performance of our proposed HO scheme is presented in Section IV. Lastly, the conclusion and future directions are given in Section VI of the paper.
%
\begin{table*}[t!]
\centering
\begin{tabular}{|c|ccccccccc|c|}
\hline


                     



\textbf{Reference} & \textbf{Single BS} & \textbf{Multiple BS} & \textbf{Single RIS} & \textbf{Multiple RIS} & \textbf{HO} & \textbf{Blockage} & \textbf{Sig} & \textbf{NF/FF}\\ \hline

{\cite{10558216}}  &    &    \checkmark        & \checkmark &            & \checkmark &             &            &            \\\hline 

{\cite{you2022deploy}}            & \checkmark &    &            & \checkmark &            &  &  &            \\\hline 
                     
 {\cite{wei2023equivalent, zhang2024discrete}}            &     & \checkmark     &            & \checkmark &   \checkmark         &   &  &            \\\hline 

 {\cite{jiao2021enabling}}            &     & \checkmark     &  \checkmark          &  &   \checkmark         &  &   &  &            \\\hline 

{\cite{okaf2021analysis, 10078238}}            &     & \checkmark     &            &  &   \checkmark         &  &  \checkmark &  &            \\\hline 

 {\cite{mollel2019handover, ghosh2020analyzing}}            &     & \checkmark     &            &  &   \checkmark         &  &   & \checkmark &            \\
 \hline
Our paper &  & \checkmark & \checkmark &  &         \checkmark   &            &   \checkmark        & \checkmark \\ \hline
\end{tabular}
\caption{Comparison with existing literature in terms of BS, HO, RIS, RR, Sig (Signaling), and NN/FF}
\label{Comparison table}
\end{table*}

\begin{figure}[t!]
    \centering
    \begin{minipage}{\linewidth}
        \centering
        \includegraphics[width=\linewidth]{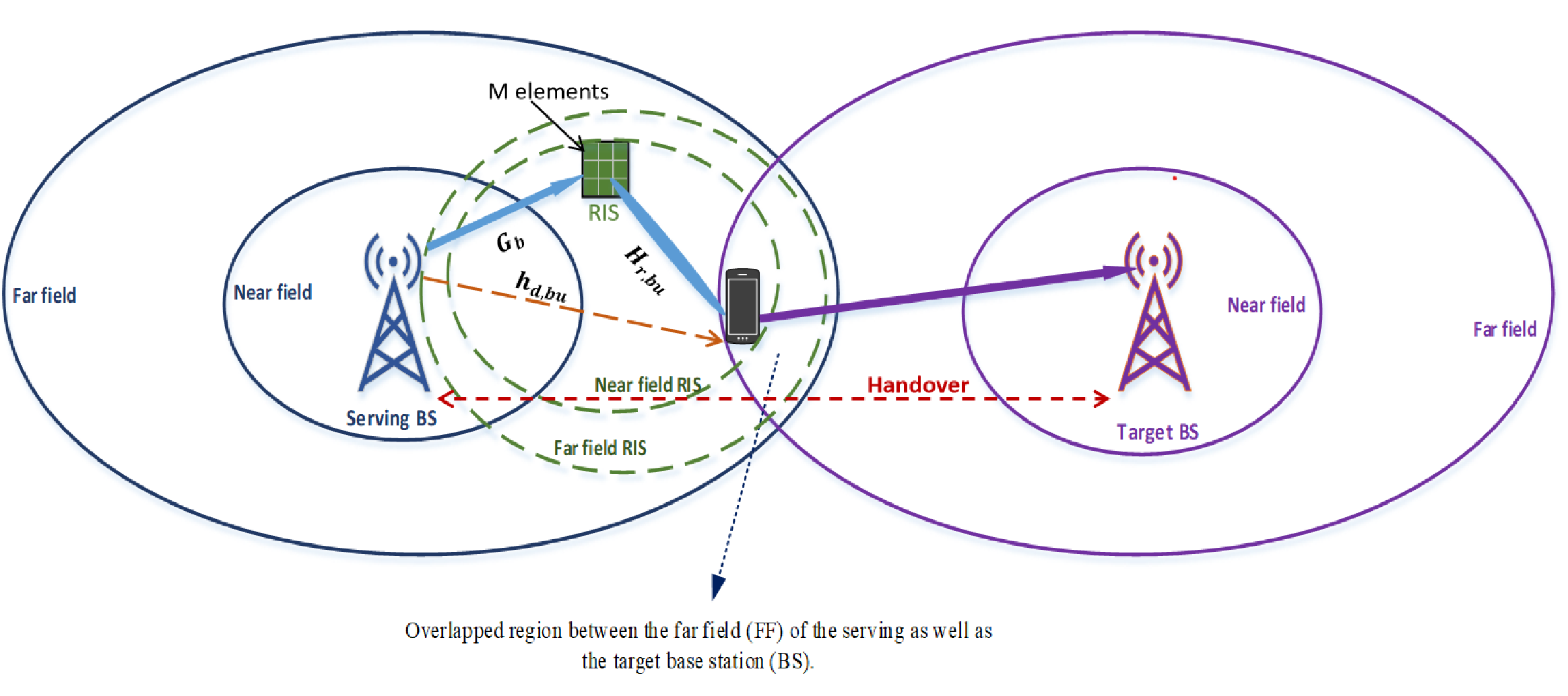}
       \caption*{(a) System model of the HO in the RIS-aided communication network in NF-FF.}
        \label{sys_model}
    \end{minipage}

    
    \begin{minipage}{\linewidth}
        \centering
        \includegraphics[width=\linewidth]{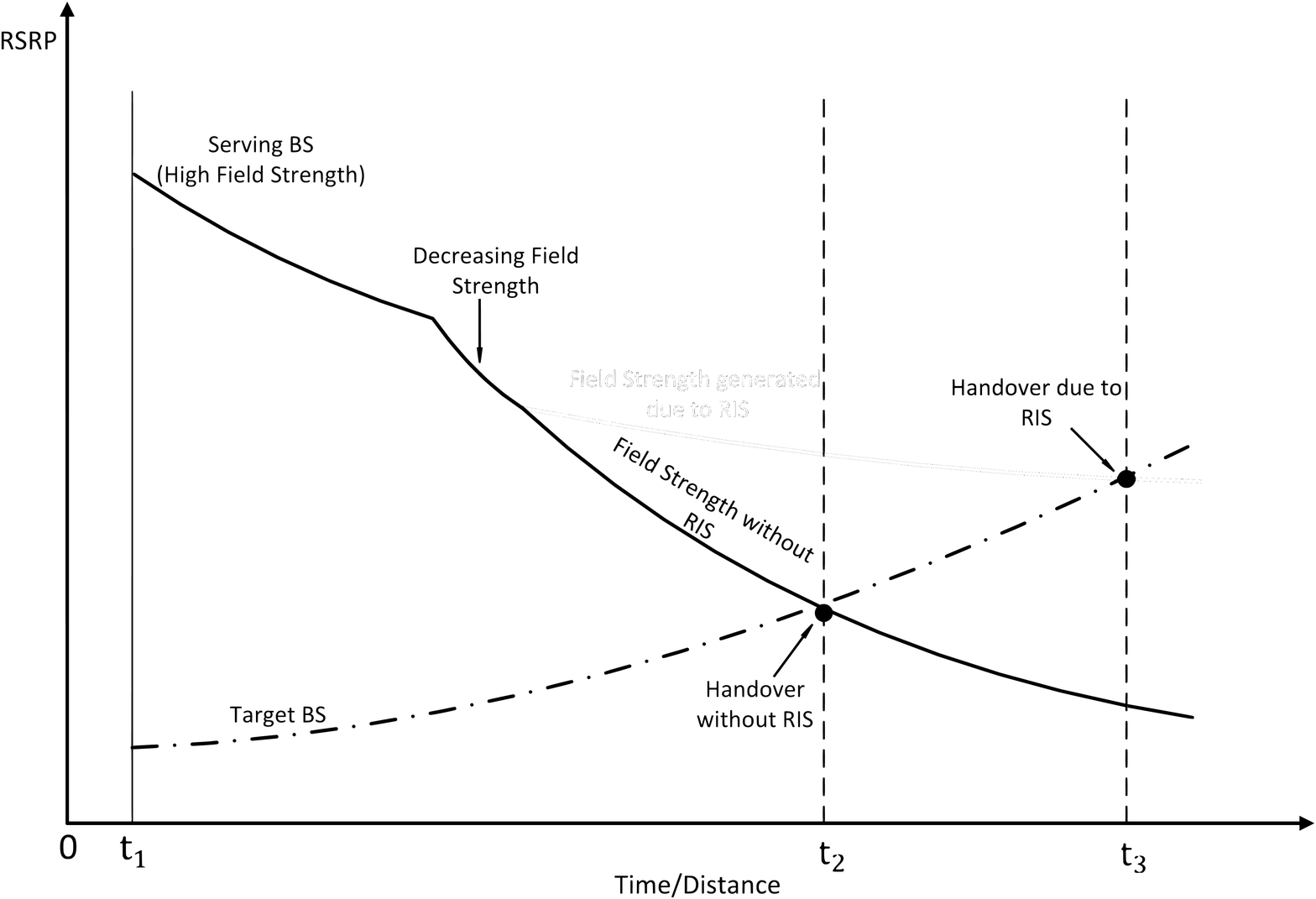}
        \caption*{(b) Illustration of RSRP values for serving and target BSs, comparing RIS-enabled and RIS-absent configurations.}
        \label{RSRP}
    \end{minipage}
    \caption{(a) RIS reconfiguration and (b) RSRP values comparison.}
   \label{combined}
\end{figure}
\section{System model}
As shown in Figure~\ref{combined}(a), a RIS-aided communication network is considered that includes a pair of BSs- serving and target, and a user. The user is assumed to be located near the boundary of the serving BS’s coverage area where the strength of the signals from the serving BS tends to be weak. Therefore, to get back in touch with the target BS, the user leverages the reflective pathway made possible by the RIS to maintain connectivity moving from the connection region of one BS towards another. We assume that users within the same cell utilize distinct orthogonal resource blocks to eliminate intra-cell interference. Given the significant propagation path loss, the field created by the RIS dynamically adjusting the weights of its antenna elements to form targeted beams to enhance signal strength for the user to have a measured decision for an optimal HO. The activation of both RIS configurations and HO occurs when wireless path reallocation becomes essential, driven by either insufficient signal strength or the identification of viable alternative paths with adequate signal quality. The frequency of these activations is captured by what we define as the RIS configuration rate and the HO rate. 

To model these rates statistically, we represent the spatial arrangement of \acp{BS} across the area as a 2-D homogeneous Poisson Point Process (PPP) with a density $\lambda_{BS}$ and a consistent transmission power of $P_{t}$, collectively referred to as $\Phi_{b}$. In our setup, we consider the presence of a RIS located between the two BS, which serve to establish alternative signal paths in situations where the signal strength between the BSs may be weakened due to distance. We assume the RIS lies within the Fresnel zone of the target BS, with its placement adhering to a PPP distribution characterized by a density $\lambda_{RIS}$. This modeling approach provides a foundation for analyzing the dynamic interplay between the field adjustments created by RIS and HO events in optimizing network performance. 

 Each BS\footnote{Note that the BS possesses perfect information of the channel state information (CSI) for every channel involved.} comprises $M$ antenna elements, while the RIS is equipped with $N$ passive reflecting elements. Meanwhile, $\mathbf{G}_{b,r} \in \mathbb{C}^{N \times M}$, $\mathbf{h}_{r,u} \in \mathbb{C}^{N \times 1}$, $\mathbf{h}_{d,u} \in \mathbb{C}^{M \times 1}$ are considered as equivalent channels corresponding to BS-RIS, RIS-UE, and BS-UE respectively. $\mathbf{\Phi} = \operatorname{diag}\left\{\Phi_1, \Phi_2, \ldots, \Phi_N\right\}$ denotes the RIS's phase shift matrix, where $\Phi_n= \beta_n e^{j \phi_n}$ is the $n_{\text{th}}$ element of the reflection matrix ( $\phi_n \in [0, 2\pi]$ 
 and $\beta_n \in [0, 1]$). 
 
\subsection{Near-field region of the RIS:}
\begin{figure}[t!]
        \centering
        \includegraphics[width=\linewidth]{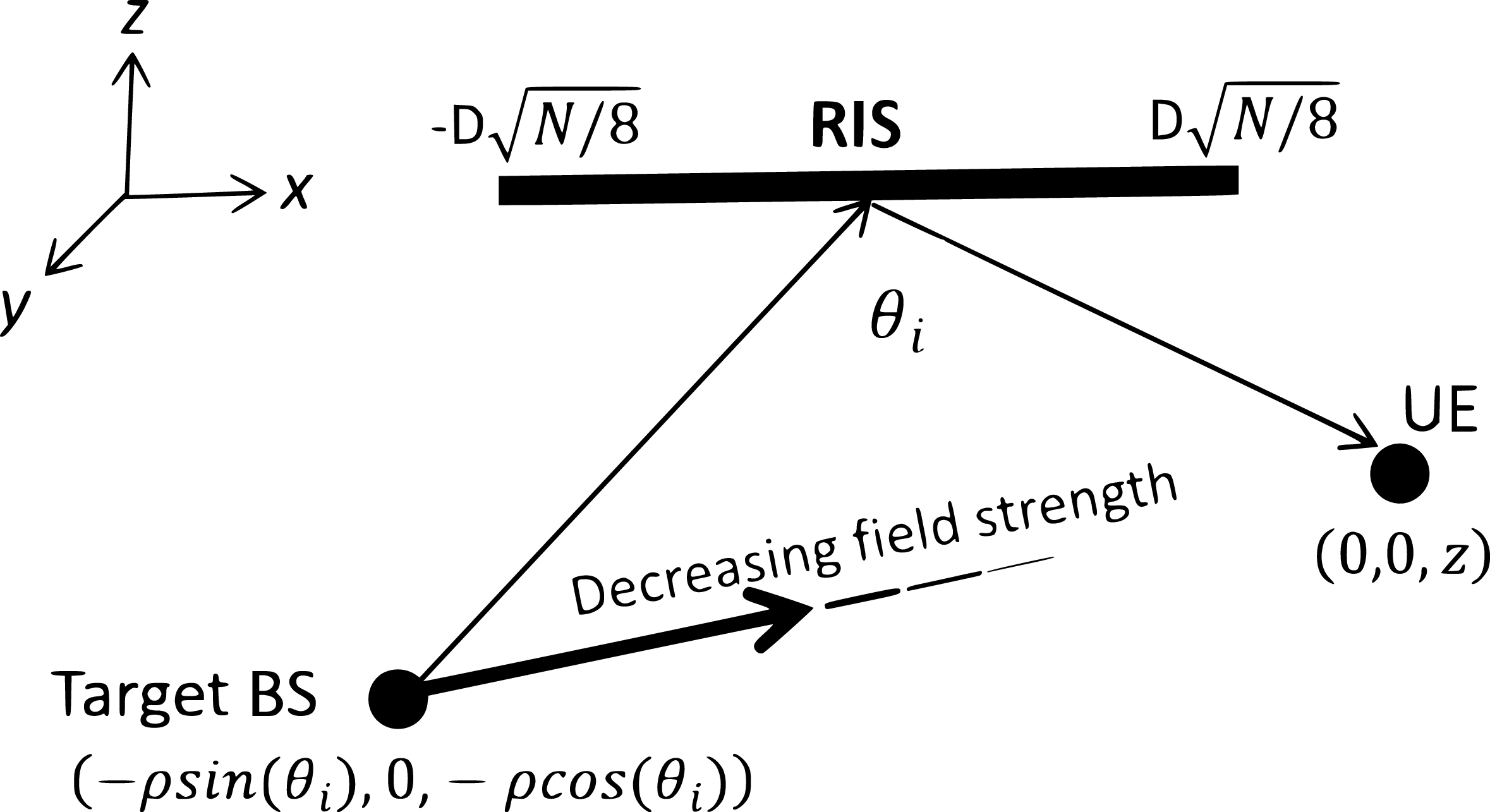}
    \caption{NF RIS reconfiguration.}
   \label{NF}
\end{figure}
Assume the following: $(-\rho \sin \theta_i, 0, \rho \cos \theta_i)$ is the coordinate of the serving \ac{BS}, $\rho$ is the distance between the serving \ac{BS} and the RIS center, $\theta_i$ is the angle of incidence in the $xz$-plane, and $(0, 0, z)$ is the location of the receiving antenna. Also, assume that there are no alternate propagation channels between the serving BS and UE and that there is an unobstructed LoS channel-links i.e., RIS-UE and RIS-serving BS. To ensure that the serving BS is in the FF area of the RIS, the distance $\rho$ is considered sufficiently large. As a result of the FF approximation, the serving BS produces a plane wave-like signal with an electric field strength $E_i$ that is polarized along the $y$-axis. The electric field's distribution throughout the RIS surface can be expressed as
\begin{align}\label{eqn_1}
  E_t(x, y) &= \frac{E_i}{2\sqrt{\pi\rho}} \exp\left\{ \frac{-2j\pi}{\lambda} (\rho + \sin(\theta_i)x)\right\},
\end{align}
where $\lambda$ is the wavelength.

However, the electric field at the UE caused by the reflected signal coming from the RIS is given by:
\begin{align}\label{eqn_2}
    E(x, y) = \frac{E_0}{2} \sqrt{\frac{z\left(x^2 + z^2\right)}{\pi\left( x^2 + y^2 + z^2 \right)^{5/2}}}* e^{-\frac{j\pi}{2} \sqrt{x^2 + y^2 + z^2}}.
\end{align}

\begin{proposition}
Assume that the Fraunhofer distance from the serving BS and the Fresnel region are characterized as $d_{\text{F}} = \frac{2D^2}{\lambda}$ and $d_F \geq 1.2D $, respectively, where $D$ is the aperture diameter. Therefore, the signal from the RIS behaves as if it were from the serving BS and therefore it can be expressed as
\begin{align}\label{prop1_eq3}
    &\left( \frac{2}{D^2 N} \right)^2 \left| e^{-j \frac{2\pi}{\lambda} z} \int\limits_{-D/\sqrt{\frac{N}{8}}}^{D/\sqrt{\frac{N}{8}}} \int\limits_{-D/\sqrt{\frac{N}{8}}}^{D/\sqrt{\frac{N}{8}}} e^{-j \frac{2\pi}{\lambda} \left( \frac{x^2 + y^2}{2F} \right)} \, dx \, dy \right|^2 \notag \\
    &= \left( \frac{8 F}{d_{\text{FA}}} \right)^2 \left\{ C^2 \left( \frac{d_{\text{FA}}}{8 F} \right) + S^2 \left( \frac{d_{\text{FA}}}{8 F} \right) \right\}^{2},
\end{align}
where $d_{FA}=Nd_{F}$ and $ F = \frac{F_z}{|F_z - z|}$ is the focal point deviation. 
\end{proposition}
\begin{proof}
    Refer Appendix A.
\end{proof}

Therefore, the compound channel link i.e., serving BS to the UE via the RIS element $(r,c)$\footnote{Here, $r \in \{1, \ldots, \sqrt{N}\}$ and $c \in \{1, \ldots, \sqrt{N}\}$ indicate the row and the column index in the $x$ and $y$ axes respectively. Therefore, the position of the antenna with index $(r,c)$ centered at $(\bar{x}_r, \bar{y}_c, 0)$ can be given as $\bar{x}_r = \left( r - \frac{\sqrt{N} + 1}{2} \right) \frac{D}{\sqrt{2}}$ and
$\bar{y}_c = \left( c - \frac{\sqrt{N} + 1}{2} \right) \frac{D}{\sqrt{2}}$ respectively.} can be expressed as
\begin{align}\label{eqn_4}
    \mathbf{h}_{r,c} &=  \frac{\sqrt{2}}{D} \left\{e^{-j\varphi_{r,c}} \int\limits_{A_{r,c}} \frac{E_t(x, y)}{E_i} \tilde{E}(x, y) dx dy\right\},
\end{align}
where $\varphi_{r,c} \in [0, 2\pi)$ denotes the phase shift induced by the RIS element.

Therefore, the overall received signal at the UE after reflecting from the RIS can now be expressed as:
\begin{align}\label{eqn_5}
    \mathbf{y}_{0} = &\sum_{i=1}^{N} \mathbf{h}_{(r,c),i} s_{u} + n_{0},
\end{align}
where $s_u$, and $n_{0}\sim \mathcal{CN}\left(0, \sigma^{2} \right)$ are the transmitted symbol, beamforming vector for UE, and additive white Gaussian noise (AWGN) respectively. Now, let us split the channel gain of the compound channel to its individual gains, i.e., $\mathcal{G}_{br,i}\sim \mathcal{N}(\mu_1, \sigma_1^2)$ and $\mathcal{G}_{ru,i}\sim \mathcal{N}(\mu_2, \sigma_2^2)$ as the gains of the channel links i.e., serving-BS to RIS and RIS to UE respectively, with $\left(\mu_{1}, \mu_{2}\right)$ being the means, and $\left(\sigma_{1}^{2}, \sigma_{2}^{2}\right)$ the variances, and all RIS elements being considered as independent and identically distributed (i.i.d.). Assuming perfect channel state information (CSI), the received signal can now be simplified as:
\begin{align}\label{eqn_6}
    \mathbf{y}_{0} = \sum_{i=1}^{N}\mathcal{G}_{br,i}\mathcal{G}_{ru,i} s_{u} + n_{0}.
\end{align}

The instantaneous SNR can therefore be expressed as
\begin{align}\label{eqn_7}
    \gamma = \frac{P_{0}\left| \sum\limits_{i=1}^{N} \mathcal{G}_{br,i}\mathcal{G}_{ru,i} \right|^2}{\sigma_n^2} = \bar{\gamma}\left( \sum_{i=1}^{N} \mathcal{G}_{bru,i} \right)^2,
\end{align}
where $\mathcal{G}_{bru} = \mathcal{G}_{br}\mathcal{G}_{ru}$, $ \bar{\gamma} = \frac{P_t}{\sigma_n^2} $ is the average SNR, and $P_{0}$ is the transmitted power. 

As large number of reflecting elements, i.e., $(N>>1)$ is considered for the creation of NF region, therefore by applying the central limit theorem, we can approximate the combined gain as $\mathcal{G} \approx \mathcal{N}(\mu_{\mathcal{G}}, \sigma_{\mathcal{G}}^2)$, where:
\begin{align}\label{eqn_8}
    \mu_{\mathcal{G}} &= N \mu_1 \mu_2, \\
    \sigma_{\mathcal{G}}^2 &= N \left[ (\sigma_1^2 + \mu_1^2)(\sigma_2^2 + \mu_2^2) - (\mu_1 \mu_2)^2 \right].
\end{align}
Therefore, the probability density function (PDF) of $\mathcal{G}$ can be expressed as
\begin{align}\label{eqn_9}
    f_{\mathcal{G}}(g) = \frac{1}{\sqrt{2\pi N \sigma_{\mathcal{G}}^2}} \exp\left( -\frac{(g - N \mu_1 \mu_2)^2}{2 N \sigma_{\mathcal{G}}^2} \right).
\end{align}
Since $\gamma = \bar{\gamma} \mathcal{G}^2$ and let $Z = \mathcal{G}^2$, therefore the PDF can be rewritten as
\begin{align}\label{eqn_10}
    f_Z(z) &= \frac{1}{\sqrt{2\pi N \sigma_{\mathcal{G}}^2 z}} \left[ \exp\left( -\frac{(\sqrt{z} - N \mu_1 \mu_2)^2}{2 N \sigma_{\mathcal{G}}^2} \right) \right.\nonumber\\
   &+  \left.\exp\left( -\frac{(\sqrt{z} + N \mu_1 \mu_2)^2}{2 N \sigma_{\mathcal{G}}^2} \right) \right], \quad z \geq 0.
\end{align}
Therefore, the PDF of $\gamma$ can now be expressed as
\begin{align}\label{eqn_11}
    &f_{\gamma}(\gamma) = \nonumber\\
    &\frac{1}{\sqrt{2\pi N \sigma_{\mathcal{G}}^2 \gamma \bar{\gamma} }} \left[ \exp\left( -\frac{\left( \sqrt{\frac{\gamma}{\bar{\gamma} }} - N \mu_1 \mu_2 \right)^2}{2 N \sigma_{\mathcal{G}}^2} \right) \right.\nonumber\\
    &+\left. \exp\left( -\frac{\left( \sqrt{\frac{\gamma}{\bar{\gamma} }} + N \mu_1 \mu_2 \right)^2}{2 N \sigma_{\mathcal{G}}^2} \right) \right], \quad \gamma \geq 0.
\end{align}
\subsection{Average Bit Error Rate (BER)}
Employing the expression obtained in \eqref{eqn_11}, the average BER can be defined as
\begin{align}\label{eqn_12}
    \text{BER} = \int\limits_0^{\infty} \frac{1}{2} \text{erfc}\left( \frac{\sqrt{\gamma}}{2} \right) f_{\gamma}(\gamma) \, d\gamma.
\end{align}
Using the change of variables $t = \sqrt{\gamma}$, so $\gamma = t^2$, $d\gamma = 2t \, dt$, we get
\begin{align}\label{eqn_13}
     \text{BER} &= \int_0^{\infty} \frac{\text{erfc}\left( \frac{t}{2} \right)}{\sqrt{2\pi N \sigma_{\mathcal{G}}^2 \bar{\gamma}}} \left[ \exp\left( -\frac{\left( \frac{t}{\sqrt{\bar{\gamma}}} - N \mu_1 \mu_2 \right)^2}{2 N \sigma_{\mathcal{G}}^2} \right) \right.\nonumber\\
    &+\left. \exp\left( -\frac{\left( \frac{t}{\sqrt{\bar{\gamma}}} + N \mu_1 \mu_2 \right)^2}{2 N \sigma_{\mathcal{G}}^2} \right) \right] \, dt.
\end{align}

\subsection{Outage Probability}
The outage probability is now be given as
\begin{align}\label{eqn_14}
    P_{\text{out}} &= P_r(\gamma \leq \gamma_{th}) = \int_0^{\gamma_{th}} f_{\gamma}(\gamma) \, d\gamma\nonumber\\
    &=\int_0^{\sqrt{\frac{\gamma_{th}}{\bar{\gamma}}}} \frac{2}{\sqrt{2\pi N \sigma_{\mathcal{G}}^2}} \left[ \exp\left( -\frac{(t - N \mu_1 \mu_2)^2}{2 N \sigma_{\mathcal{G}}^2} \right) \right.\nonumber\\
    &+\left. \exp\left( -\frac{(t + N \mu_1 \mu_2)^2}{2 N \sigma_{\mathcal{G}}^2} \right) \right] \, dt.
\end{align}
\subsection{Channel Capacity}

Using \eqref{eqn_12}, the channel capacity can be expressed as
\begin{align}\label{eqn_15}
    C &= \frac{1}{\ln 2} \int_0^{\infty} \ln(1 + \gamma) f_{\gamma}(\gamma) \, d\gamma\nonumber\\
    &= \frac{1}{\ln 2} \int_0^{\infty} \ln(1 + t^2 \bar{\gamma}) \nonumber\\
    &\times \frac{2}{\sqrt{2\pi N \sigma_{\mathcal{G}}^2}} \left[ \exp\left( -\frac{(t - N \mu_1 \mu_2)^2}{2 N \sigma_{\mathcal{G}}^2} \right) \right.\nonumber\\
    &+\left. \exp\left( -\frac{(t + N \mu_1 \mu_2)^2}{2 N \sigma_{\mathcal{G}}^2} \right) \right] \, dt.
\end{align}
\subsection{Remarks}

\textit{The non-central chi-squared distribution of $\mathcal{G}^2$ resulted in complex integrals for BER, outage probability, and channel capacity. Therefore, Monte Carlo simulations or numerical methods are recommended for their evaluation. However, for large $N$, the Central Limit Theorem (CLT) provides a Gaussian approximation of $\mathcal{G}$, which simplifies portions of the analysis.}
\section{Analysis of Handover Process}

In this section, we put forward a new algorithm (Algorithm~\ref{algo}) for maneuvering HO within the network configuration as explained in Figure~\ref{combined}(a), that takes into account multiple HO types, including hard HO (HHO), soft HO (SHO), RIS-aided cell breathing (RIS-CB), and RIS-aided Ping-Pong avoidance (RIS-PP). In order to provide seamless connectivity, the algorithm employs NF-RIS, for which the HO decision process utilizes the BER parameter to select the optimal communication link prior to the potential link failure, as is supposed to happen in the FF scenario, that too without the assistance of RIS. By calculating HO probabilities for both hard and soft HO, the algorithm emphasizes the edge of NF-RIS-assisted networks over traditional non-RIS-based HO approaches. We consider a UE positioned at $P_{UE} = (x, y, z)$, initially connected to a serving BS through link $l_s$ in the FF, maintaining a BER below a QoS threshold of $10^{-3}$. The link is periodically assessed to ensure compliance with QoS requirements. When the BER from the serving BS exceeds this threshold, the UE scans for alternative BSs (e.g., a target BS) to enable a connectivity switch upon entering the Rayleigh distance of the target. We categorize links as either direct (LoS) or non-direct (NLoS). The UE might receive the channel links from so many nearby BS\footnote{Though in our case we will restrict to one target BS. However, in reality there might be many target BSs in the vicinity of the UE.}, therefore the set of direct links available to the UE at $P_{UE}$ can be defined as:
\begin{align}\label{eqn_16}
    \mathcal{F}_{dir}^{P_{UE}} = \{ p_{d,1}, p_{d,2}, \dots, p_{d,i}, \dots, p_{d,N} \},
\end{align}
where the BER of each connection surpasses the sensitivity threshold of the UE,  after being calculated at the UE.
\begin{algorithm}
\caption{Algorithm for the NF-RIS-aided HO management}
\label{algo}
\begin{algorithmic}[1]
    \Require
        \State $P_{UE} = (x, y, z)$ \Comment{Location of the UE }
        \State $\mathcal{F}_{dir} = \{p_{d,1}, \ldots, p_{d,i}, \ldots, p_{d,N}\} $ \Comment{The set of direct links existing}
        \State $\mathcal{F}_{nodir} = \{p_{nd,1}, \ldots, p_{nd,j}, \ldots, p_{nd,M}\} $ \Comment{The set of non-direct links existing}
        \State $ T_{hh} $ \Comment{The threshold value of HHO}
        \State $ T_{hs} $, where $ T_{hs} < T_{hh} $ \Comment{The threshold value of SHO}
        \State $ p_s \in \mathcal{F}_{dir} \cup \mathcal{F}_{nodir} $ \Comment{The link for the serving BS}
        \State $ \rho_p $ \Comment{The active connections on $ p_s $}
        \State $ \lambda $ \Comment{The threshold value for $ \rho_p $}
        \State $ \epsilon $ \Comment{ BER margin}
    \Ensure
        \State $ h \in \{HHO, SHO, RIS-CB, RIS-PP\} $ \Comment{Handover execution}
    \Procedure{Handover Execution}{}
        \If{$ BER_{p_s} \geq T_{hh} $}
            \For{$ i = 1 $ to $ N $}
                \If{$ BER_{p_{i}} < BER_{p_{s}} $ and $ BER_{p_{i}} \leq T_{hh} $}
                    \State $ h \gets HHO $ \Comment{Execute hard handover}
                \EndIf
            \EndFor
            \For{$ i = 1 $ to $ N $}
                \If{$ BER_{d,i} < Th_s $}
                    \State $ h \gets SHO $ \Comment{Execute soft handover}
                \EndIf
            \EndFor
        \EndIf
        \If{$ T_{hh} - \epsilon < BER_{p_{s}} \leq T_{hh} + \epsilon $}
            \For{$ j = 1 $ to $ M $}
                \If{$ BER_{p_{j}} < T_{hh} $}
                    \State $ h \gets RIS-PP $ \Comment{Perform RIS-aided PP avoidance handover}
                \EndIf
            \EndFor
        \EndIf
        \If{$ \rho_{p_{s}} > \lambda $}
            \For{$ j = 1 $ to $ M $}
                \If{$ BER_{p_{j}} < T_{hh} $}
                    \State $ h \gets RIS-CB $ \Comment{Perform RIS-aided CB HO}
                \EndIf
            \EndFor
        \EndIf
    \EndProcedure
\end{algorithmic}
\end{algorithm}

When the performance of the serving link $ p_s $ degrades due to decreasing signal strength in the FF, the algorithm concentrates on maintaining connectivity with the serving BS by switching to the NF-RIS-aided route. This approach eliminates the need for additional UE authentication and association overhead, as the serving BS remains unchanged, with only the propagation path modified. The set of available NLoS links via RIS is defined as:
\begin{align}\label{eqn_17}
    \mathcal{F}_{nodir}^{P_{UE}} = \{ p_{nd,1}, p_{nd,2}, \dots, p_{nd,j}, \dots, p_{nd,M} \}.
\end{align}
The probability of initiating a HHO from the serving BS (via $ p_s $) to a target BS through the $ i $-th direct link in $ \mathcal{F}_{dir}^{P_{UE}}$ is measured based on average BER, that is given as
\begin{align}\label{eqn_18}
    P_{HHO, p_i} &= \Pr \left(  \text{BER}_{p_s} \geq T_{\text{hh}} \bigcup  \text{BER}_{p_i} <  \text{BER}_{p_s} \right),
\end{align}
where $\text{BER}_{p_s}$ and $\text{BER}_{p_i}$ are the BERs of the current link to the serving BS via path $p_s$ and the potential link to the target BS via path $p_i$, respectively. While $T_{\text{hh}}$ is the predefined BER threshold for triggering a hard HO, usually considered to be $10^{-3}$. The operator $\bigcup$ indicates that the HHO is initiated if either of the condition or both is satisfied. The first condition, i.e., $\text{BER}_{p_s} \geq T_{\text{hh}}$, states that the serving link’s quality has degraded beyond an acceptable QoS threshold, while the second condition $\text{BER}_{p_i} < \text{BER}_{p_s}$, ensures that the target BS offers superior link quality. Therefore, using the principle of inclusion-exclusion, \eqref{eqn_18} can be written as
\begin{align}\label{eqn_19}
    P_{HHO, p_i} &=\Pr \left( \text{BER}_{p_s} \geq T_{\text{hh}} \right) + \Pr \left( \text{BER}_{p_i} < \text{BER}_{p_s} \right) \nonumber\\ 
    &- \Pr \left( \text{BER}_{p_s} \geq T_{\text{hh}}, \text{BER}_{p_i} < \text{BER}_{p_s} \right).
\end{align}
Once the $ P_{HHO, p_i}$ is formulated, the HHO is then implemented towards the link with the highest probability, which is given as
\begin{align}\label{eqn_20}
    \max_{i \in \mathcal{F}_{dir}^{P_{UE}}} P_{HHO, \mathcal{F}_{P_{UE}}^{dir}}.
\end{align}
Similarly, for the case of SHO, the probability relates to the serving BS to the target BS via the $ i $-th direct link with $T_{\text{hs}}$ being the threshold of the BER for the SHO, can be defined as
\begin{align}\label{eqn_21}
    P_{SHO, p_i} &= \Pr \left( T_{hs} \leq BER_{p_s} < T_{hh} \bigcup BER_{p_i} < T_{hs} \right)\\
    &= \Pr \left( T_{\text{hs}} \leq \text{BER}_{p_s} < T_{\text{hh}} \right) + \Pr \left( \text{BER}_{p_i} < T_{\text{hs}} \right) \nonumber\\
    &- \Pr \left( T_{\text{hs}} \leq \text{BER}_{p_s} < T_{\text{hh}}, \text{BER}_{p_i} < T_{\text{hs}} \right).
\end{align}
As explained for the HHO case, the first condition here i.e., $T_{\text{hs}} \leq \text{BER}_{p_s} < T_{\text{hh}}$, indicates that the link quality of the serving BS is degraded but remains below the HHO threshold, while the second condition, i.e., $\text{BER}_{p_i} < T_{\text{hs}}$, ensures that the target BS provides a high-quality link once the UE moves away from both the vicinity region of the serving BS as well as NF-RIS. Here, the link from the serving BS is preserved, and the UE receives data from both the BSs. The soft HO is then carried out towards the link with the highest probability, which is given as
\begin{align}\label{eqn_22}
    \max_{i \in \mathcal{F}_{dir}^{P_{UE}}} P_{SHO, \mathcal{F}_{dir}^{P_{UE}}}.
\end{align}
Let us now look into the cases where we can leverage the benefits of NF-RIS-aided in terms of HO. For this, we extend the formulated probabilities to include both $\mathcal{F}_{dir}$ and $\mathcal{F}_{nodir}$. By utilizing the RIS in NF to extend the Fresnel region or in other words we can say that, to extend the coverage of the original serving BS, the algorithm maintains the active connection to the serving BS while optimizing the propagation path. The effectiveness of the proposed method, where the NF-RIS improves the HO, we will consider two cases i.e., cell breathing (CB) and \ac{HOPP} reduction. The HO probability for the CB as well as the for the PP avoidance for the $ j $-th NLoS link are given as
\begin{align}\label{eqn_23}
    &P_{RIS, p_{j}}^{CB} = \Pr \left( \rho_{p_s} > \lambda \bigcup BER_{p_{j}} < T_{hh} \right), and\\ 
    & P_{RIS,p_{j}}^{PP} = \Pr \left( (T_{hh} - \epsilon < BER_{p_s} < T_{hh} + \epsilon) \right.\nonumber\\
    &\bigcup\left.(T_{hh} - \epsilon < BER_{p_{j}} < T_{hh} + \epsilon) \right.\nonumber\\
    &\bigcup\left. (BER_{p_{j}} < T_{hh}) \right),
\end{align}

where $ \rho_{p_s} $ represents the number of active connections on the serving BS, $\epsilon $ is a small margin, and $ \lambda $ is the threshold triggering the CB mechanism. 

When the value of the traffic load goes past the threshold $\lambda$, and the BER of the $j$-th RIS-aided link falls below a predefined threshold $T_{hh}$, then a HO is triggered to that link, that occurs while maintaining the serving BS, while redirecting the signal through the RIS-aided channel link. The selection of the optimal path is determined by maximizing HO probability in the set $\mathcal{F}_{P_{UE}}^{\text{nodir}}$. Furthermore, in the case of HO-PP, the transition from the serving direct link with probability $p_s$ to the $j$-th RIS-aided link is facilitated if the BERs of both the direct and RIS-aided channel paths fall within the range $ T_{hh} \pm \epsilon $, and the BER of the $ j $-th link is still below $T_{hh}$. The HO is then executed to the link with the highest probability within the set $\mathcal{F}_{P_{UE}}^{\text{nodir}}$, ensuring that the serving BS is preserved while optimizing the RIS-aided channel path.
In addition to it, it is crucial to note that in the case of NF-RIS-aided PP avoidance HO probability, the serving BS remains constant throughout the RIS-aided link, that forms a key aspect of the RIS-aided HO procedure, which aims to preserve the serving BS while redirecting the signal propagation via alternative links facilitated by the RIS.
\section{Numerical results}
\begin{figure}[t!]
        \centering
        \includegraphics[width=\linewidth]{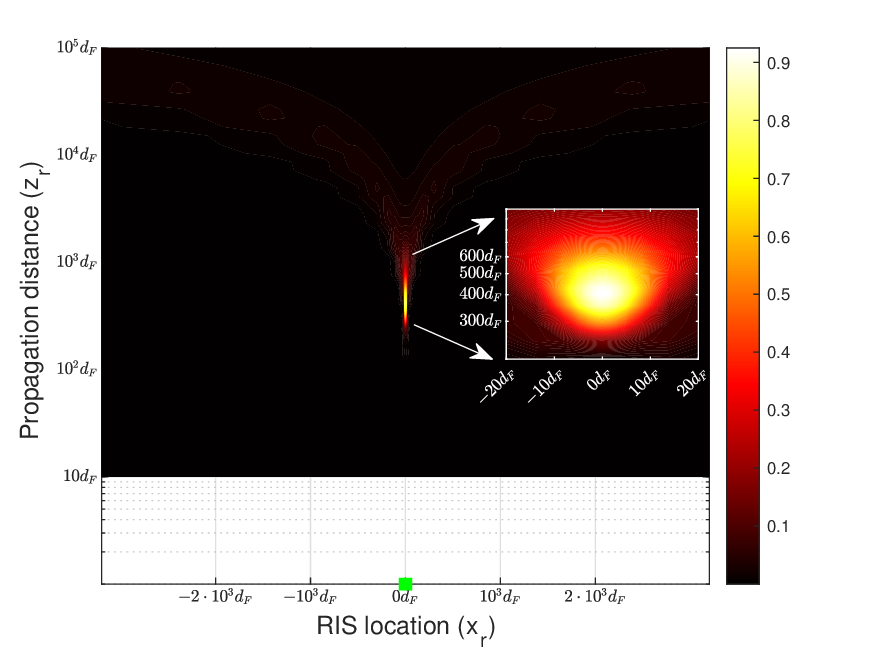}
    \caption{Signal beam directed towards the UE by the NF-RIS.}
   \label{NF beam}
\end{figure}
Figure~\ref{NF beam} is a visual representation of how the RIS gains can focus the signal beam precisely to a specified direction. The heatmap is obtained after the RIS is placed at the location i.e., $0 d_{F}$ in the $x_{r}$ direction, whereas the UE is located somewhere in the \(x\text{-}z\)-plane.  The incident signal is directed towards the UE placement site by the RIS phase-shifts. According to the zoomed map, the focal plane's 3 dB bandwidth along the $x_r$-axis and the $z_r$-axis's 3 dB beam depth are achieved. These results are extremely narrow, suggesting that the UE has more coverage to delay the HO process than in the FF scenario.

\begin{figure}[t!]
        \centering
        \includegraphics[width=\linewidth]{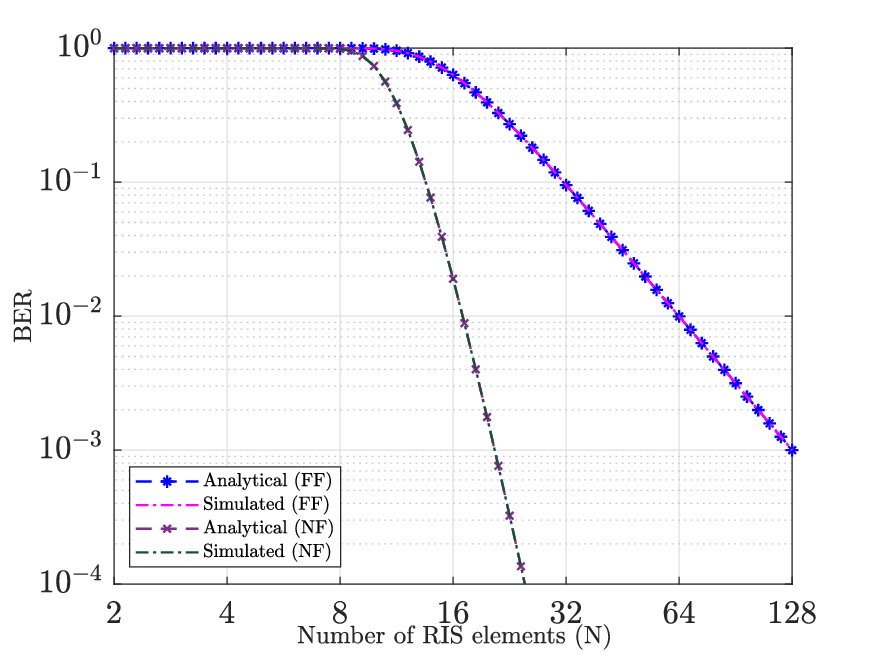}
    \caption{BER versus number of RIS-elements.}
   \label{BER}
\end{figure}

In what follows, we investigate the three important performance indicators (i.e., BER, OP, and channel capacity) of the UE under NF and FF environments. The performance trends are depicted graphically in Figures~\ref{BER}, \ref{OP} and \ref{Cap}, where the analysis has been conducted. As we can see from these figures, a significant improvement in all the KPIs is observed as the number of RIS elements increases, especially in the NF case. This gain is due to the enhanced focusing capability of RIS, which can make the signal beam optimal. Moreover, a close correlation can be observed between the analytical models and simulation results, which is indicated by the coinciding lines overlaid onto the data. This coincidence confirms the trustworthiness of the theoretical model used in the prediction of system's responses. The enhanced performance in the NF area means that the signal beam generated by RIS has enough robustness to expand the effective coverage region for the UE to a distance level. This extended horizon is an important factor in the delay of the HO process towards the target BS and indeed the unwanted inter-cell HOs occurrence, and ultimately reducing the risk of PP.

\begin{figure}[t!]
        \centering
        \includegraphics[width=\linewidth]{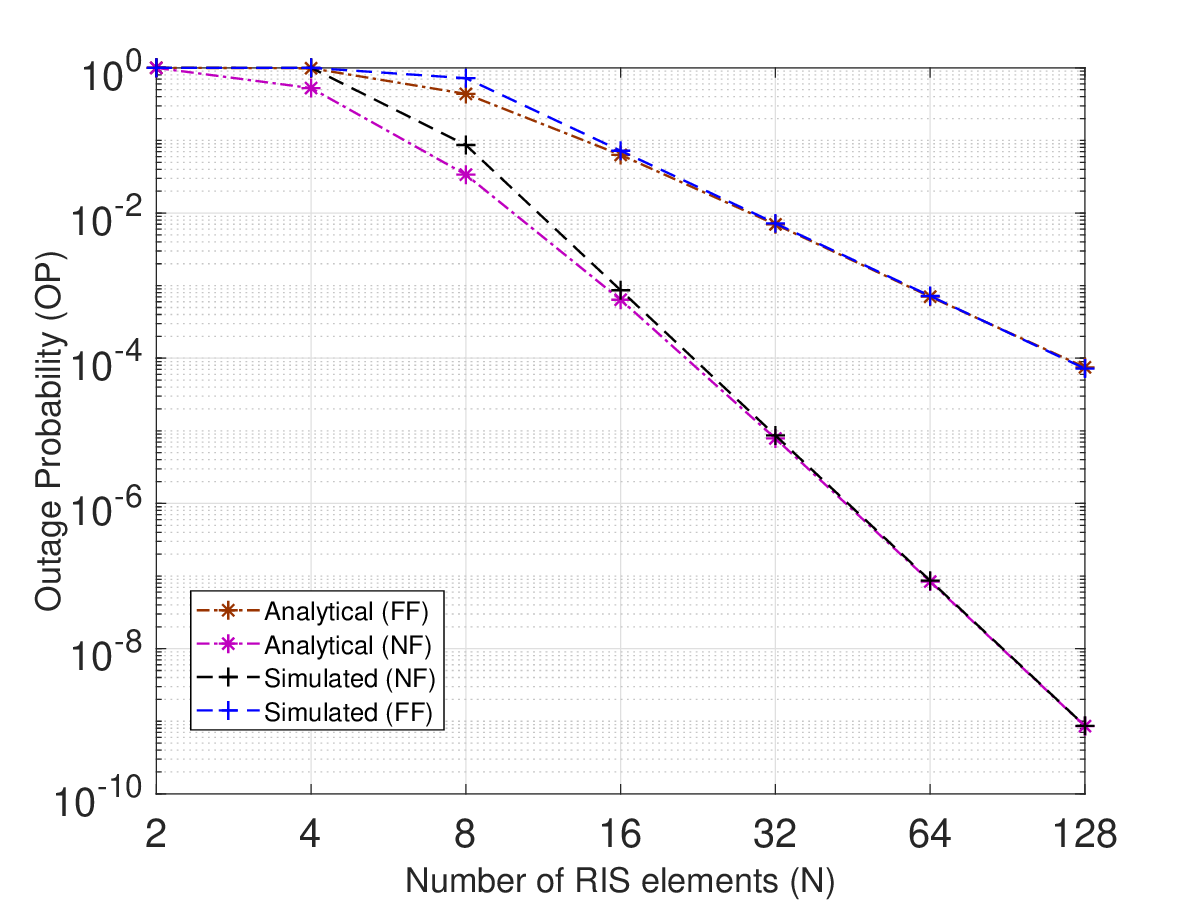}
    \caption{Outage Probability versus number of RIS-elements.}
   \label{OP}
\end{figure}
\begin{figure}[t!]
        \centering
        \includegraphics[width=\linewidth]{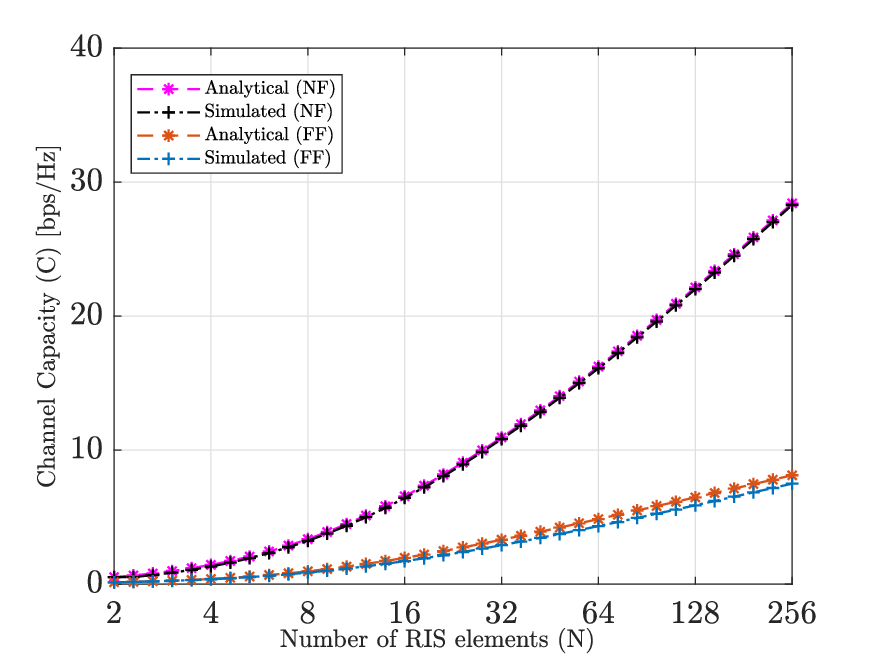}
    \caption{Channel capacity versus number of RIS-elements.}
   \label{Cap}
\end{figure}
\begin{figure}[t!]
        \centering
        \includegraphics[width=\linewidth]{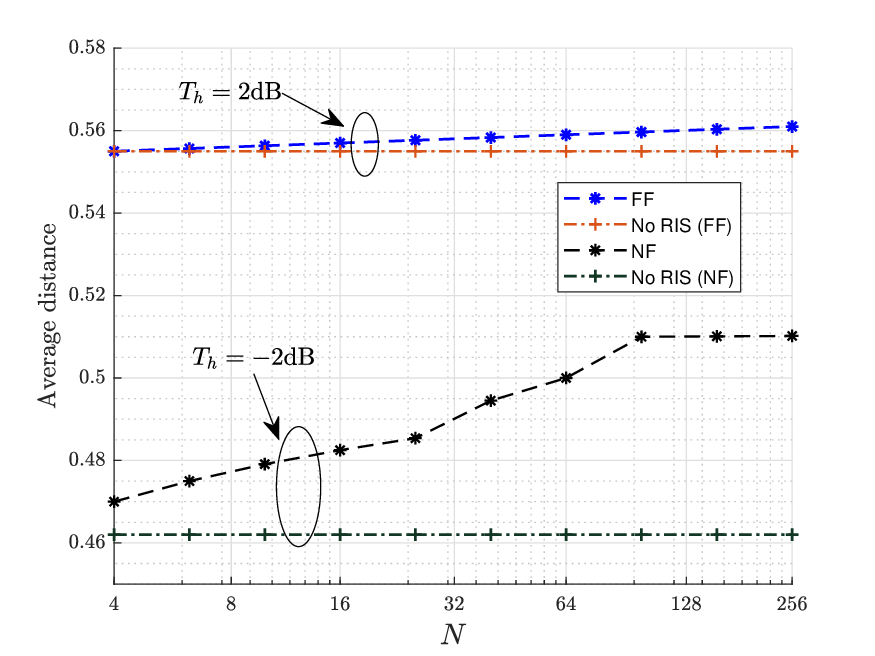}
    \caption{Average distance from the starting point to the point where the HO is triggered under different HO margin versus number of RIS-elements for D=60 m.}
   \label{dis20}
\end{figure}
\begin{figure}[t!]
        \centering
        \includegraphics[width=\linewidth]{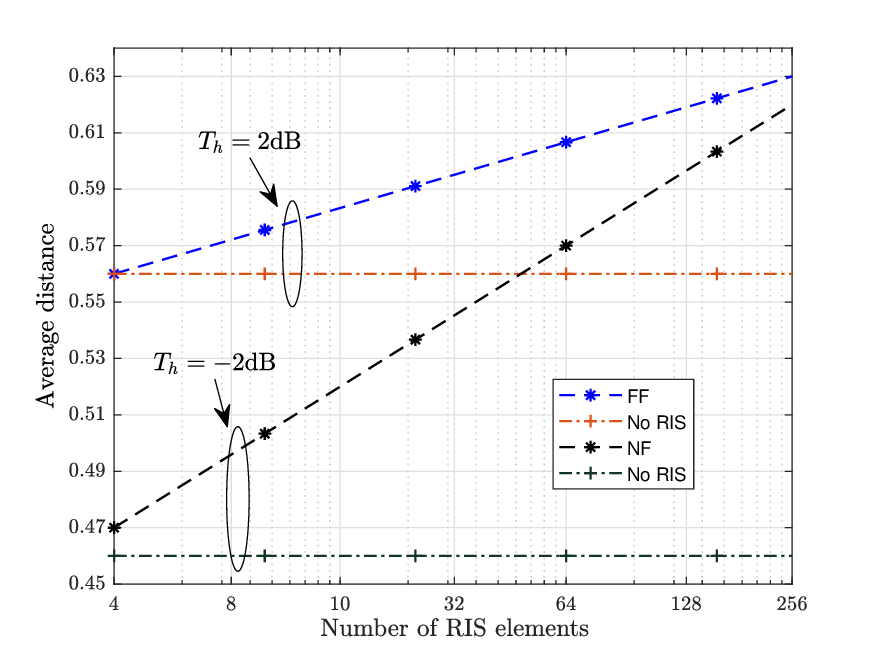}
     \caption{Average distance from the starting point to the point where the HO is triggered under different HO margin versus number of RIS-elements for D=150 m.}
   \label{dis50}
\end{figure}

After the validation of the virtual NF signal beam performance at UE, the next issue to consider is HO process. This analysis is explained by considering the plots shown in Figure~\ref{dis20} and Figure~\ref{dis50}, which show the impact of the number of RIS elements on the HO triggers positions. These numbers take into account different serving distance of RIS and HO margins, which provides an in-depth view of the dynamic change in HO policies. The HO procedure largely depends on the threshold $T_{h}$, which is the HO threshold to switch to a new BS based on the measurement report of UE. When $T_{h} < 1$, the HO trigger positions tend to move closer to the serving BS. This phenomenon highlights a crucial dependence on the threshold choice, and the RIS is more influential on the locations of HO trigger. To illustrate this impact, we assume that the number of RIS elements is $N = 100$ and the serving distance is $D = 150$ m, as shown in Figure~\ref{dis50}. In this setup, for an HO triggering threshold $T_h =-2$dB, the average distance to the HO triggering remains similar on average with about 26\% more movements compare with the reference without the RIS. This distance gain is due to the signal concentration and range expansion properties of RIS. We still have an increase in average distance, but such an increment is much smaller as compared with the case of $T_h = 2$ dB which is slightly upper by approximately 7\%, which implies the influence of the RIS decrease when the threshold value is high. Additionally, when $T_h = 2$dB, the introduction of RIS does not cause large changes of the HO position, which shows a saturation phenomenon when the threshold condition is good. It is also more complex form that the HO trigger position is not continuously highly variable when increasing numbers of RIS elements because $N$ goes past a certain value. This saturation is a function of the serving distance $D$ and the threshold $T_h$. This behavior implies that adding more RIS elements beyond the critical numbers is less effective, and the saturation threshold increases with the serving distance and the threshold $T_h$. 
\begin{figure}[t!]
        \centering
        \includegraphics[width=\linewidth]{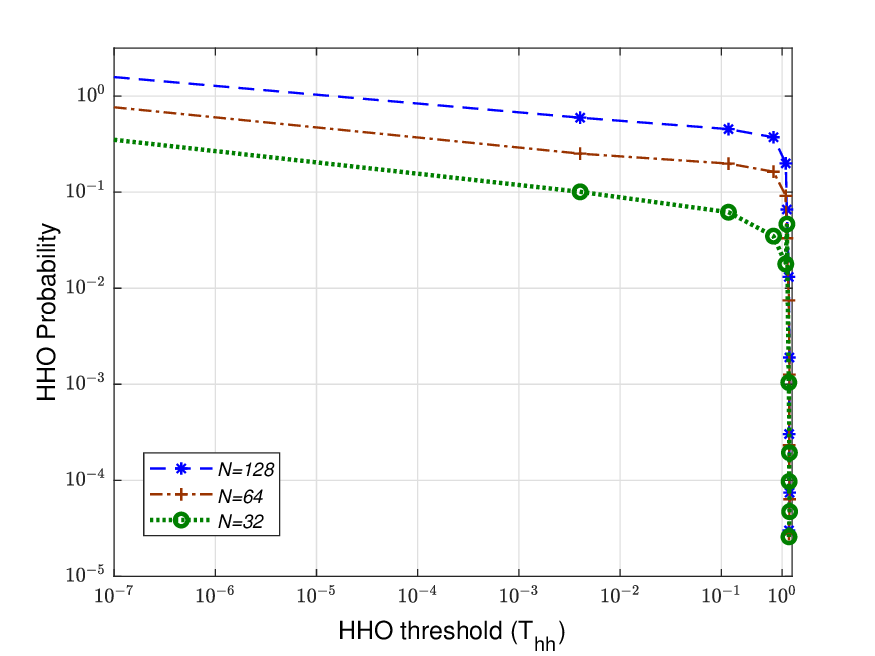}
     \caption{HHO probability versus HO threshold for varied RIS elements in NF scenario.}
   \label{HHO}
\end{figure}

The effect of changing the number of RIS on the HO probability, especially the case of a HO from a serving BS to a target BS within the virtual Fresnel zone formed by NF RIS is illustrated with the aid of Figure~\ref{HHO}. This figure shows the HO probability versus the number of RIS elements, represented by $N$, and offers valuable information about the HO performance of RIS-assisted communication links. In Figure~\ref{HHO}, we have plotted the HHO probability under different configurations of RIS elements. It is observed from the results that the HHO probability is maximized as UE moves from the serving BS to the target BS via a reflected UE link with an increasing $\mathcal{N}$. More specifically, it can be observed that the HO probability does not significantly change for different $N$ values. This shows that the additional improvement in HO performance becomes marginal when the number of RIS elements is very high. That is, as $N$ grows, the RIS-aided link performs better, but the gain of further increasing the number of elements diminishes after a certain threshold. Numerically, the gain of the RIS-aided link with larger number of elements (e.g., $N= 128$) is higher compared to that with smaller number of elements (e.g., $N = 32$), justifying the use case of increasing the number of RIS elements. The augmented RIS elements improves the performance of the beamforming and reﬂection of RIS, which in return elevates the signal quality and the link reliability towards the target BS. Therefore, a larger $N$ leads to a stronger virtual Fresnel zone which is beneficial to the HO process. A high HHO probability implies that the HO decision can be made more easily, indicating that the RIS-aided link has a more stable and stronger connection to the target BS. On the contrary, if the probability of HHO is low, the HO provides few advantages. In fact, the increase in the link quality may not reward the transition. Hence, the RIS element number optimization is of great importance for RIS-assisted wireless networks to maximize the performance of HO, especially in the near field, where signal propagation is impacted by the virtual Fresnel zone.

\begin{figure}[t!]
        \centering
        \includegraphics[width=\linewidth]{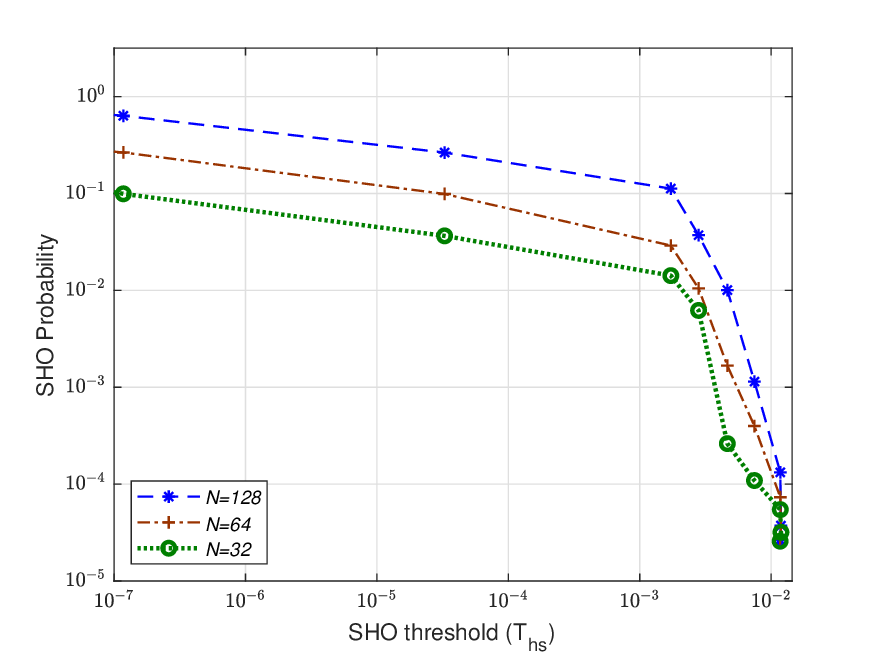}
     \caption{SHO probability versus HO threshold for varied RIS elements in NF scenario.}
   \label{SHO}
\end{figure}

The soft HO probability, which is illustrated in Figure~\ref{SHO}, is considered for the RIS-aided virtual NF scenario. The value of $T_{hs}$, the soft HO threshold, varies from $10^{-7}$ upto $10^{-2}$. This is a narrower range as compared to the hard HO threshold ($T_{hh}$) ($T_{hs} < T_{hh}$). Thus, the likelihood of performing a soft HO is lower than that of a hard HO and is limited by the operational limits of the soft HO threshold. As shown in Figure~\ref{SHO}, the highest SHO probabilities are experienced when the RIS configuration has a size of $\mathcal{N} = 128$. Consequently, lower probabilities are found for HO that are related to $\mathcal{N} = 32$. This is evolution of the beamforming in the virtual (NF) environment assisted by RIS. In particular, beamforming connections with smaller RIS, e.g., $\mathcal{N} = 32$, suffer from performance degradation compared to the ones with larger RIS, e.g., $\mathcal{N} = 128$. The better beamforming gain of large RIS configuration is beneficial to the link performance, which mitigates the requirement of HO and then decreases the corresponding soft HO probability. When the size of the RIS is large, the baseline performance of the beamforming link is strong enough. The marginal benefits of a soft HO are therefore small, leading to a low SHO probability. On the other hand, as for the links with less RIS elements, the performance deterioration becomes more obvious, resulting in the more number of SHOs to satisfy the quality of service. Note that the bias toward SHO in the case of larger number of RIS elements is clear, since the higher RIS enabled sessions result in more dependable beamforming connections. The SHO probability is also affected by the bit error rate (BER) of serving link. When the BER is above the threshold range of $[T_{hs}, T_{hh}]$, and especially when $BER < T_{hs}$, the serving link performance is high so that SHO is hardly necessary. This effect is more significant in $\mathcal{N} = 32$ and $\mathcal{N} = 64$ configurations with lower SHO probabilities due to the better BER performance of these links. In these circumstances the system chooses to maintain the nomadic link rather than to set up a HO, which has the additional effect of reducing the SHO probability.
\section{Conclusion and future directions}
This paper proposed an efficient HO management approach leveraging the virtual fresnel region created by RIS to enhance the coverage as well as providing an efficient way for optimal HO in 6G networks. By leveraging RIS to extend coverage and optimize signal paths, the proposed approach significantly reduces signaling overhead, improves energy efficiency, and minimizes unnecessary HOs. The algorithm evaluates the link quality harnessing the error performance, ensuring the seamless mobility with the extended coverage and eventually having an optimal HO, reducing the chances of \acp{HOPP}.

Future directions will focus on implementing various HO features, including low-layer-triggered mobility, to further reduce HO interruption time. Moreover, reducing energy consumption by developing more efficient HO algorithms remains a key requirement for future mobile networks.

\section*{Appendix A}
Considering the Fraunhofer distance from the serving BS to be $d_{\text{F}} = \frac{2D^2}{\lambda}$, and the Fresnel region as $ d_F \geq 1.2D $, the normalized antenna gain can be calculated using numerical integration from equation \eqref{eqn_2}, expressed as
\begin{align*}\label{eqn_A.1}
G_{\text{antenna}}&= \frac{E_{0}^{2}|h|^{2}/\eta}{\int\limits_A \left| E(x, y) \right|^2 \, dx \, dy/\eta}\nonumber\\
&= \frac{E_{0}^{2}\left|\frac{1}{E_0} \sqrt{\frac{2}{D^2}} \int\limits_A E(x, y) \, dx \, dy\right|^{2}/\eta}{\int\limits_A \left| E(x, y) \right|^2 \, dx \, dy/\eta}\nonumber\\
&=\frac{\left| \int\limits_A E(x, y) \, dx \, dy \right|^2}{\left( \frac{D^2}{2} \right) \int\limits_A \left| E(x, y) \right|^2 \, dx \, dy}. \tag{A.1}
\end{align*}
Subsequently, to account for phase variations, we incorporate a phase-shift term into the integrand of \eqref{eqn_A.1}, which is modeled as $ e^{+j \frac{2\pi}{\lambda} \left( \frac{x^2}{2F} + \frac{y^2}{2F} + \rho \sin(\theta_i) x \right)} $ and then integrated for the whole RIS. Additionally, to generalize the expression, we replace the term $ \frac{2}{D^2} $ with $ \frac{2}{D^2 N} $ to generate a continuous representation of the phase shifts induced by the RIS, as presented in \eqref{prop1_eq3}.

\printbibliography 

@article{jiao2021enabling,
  title={Enabling efficient blockage-aware handover in RIS-assisted mmWave cellular networks},
  author={Jiao, Long and Wang, Pu and Alipour-Fanid, Amir and Zeng, Huacheng and Zeng, Kai},
  journal={IEEE Transactions on Wireless Communications},
  volume={21},
  number={4},
  pages={2243--2257},
  year={2021},
  publisher={IEEE}
}

@article{wang2024wideband,
  title={Wideband beamforming for RIS assisted near-field communications},
  author={Wang, Ji and Xiao, Jian and Zou, Yixuan and Xie, Wenwu and Liu, Yuanwei},
  journal={IEEE Transactions on Wireless Communications},
  year={2024},
  publisher={IEEE}
}

@inproceedings{zhou2024near,
  title={Near-Field Multi-Beam Design for Extremely Large-Scale RIS Assisted by Multi-Start Adam Algorithm},
  author={Zhou, Kang and Zhou, Weixi and Xiao, Sa and Meng, Linggang},
  booktitle={Proceedings of the 2024 12th International Conference on Communications and Broadband Networking},
  pages={70--75},
  year={2024}
}

@article{wei2025analysis,
  title={Analysis of intelligent reflecting surface-enhanced mobility through a line-of-sight state transition model},
  author={Wei, Haoyan and Zhang, Hongtao},
  journal={IEEE Transactions on Vehicular Technology},
  year={2025},
  publisher={IEEE}
}

@article{wei2025handover,
  title={Handover Optimization for IRS-Aided Networks With Experimental Verification},
  author={Wei, Haoyan and Zhang, Hongtao},
  journal={IEEE Transactions on Wireless Communications},
  year={2025},
  publisher={IEEE}
}

@article{sun2021multi,
  title={A multi-attribute handover algorithm for QoS enhancement in ultra dense network},
  author={Sun, Kai and Yu, Jiarun and Huang, Wei and Zhang, Haijun and Leung, Victor CM},
  journal={IEEE Transactions on Vehicular Technology},
  volume={70},
  number={5},
  pages={4557--4568},
  year={2021},
  publisher={IEEE}
}

@article{gannapathy2023smart,
  title={{A smart handover strategy for 5G mmWave dual connectivity networks}},
  author={Gannapathy, Vigneswara Rao and Nordin, Rosdiadee and Abdullah, Nor Fadzilah and Abu-Samah, Asma},
  journal={IEEE Access},
  volume={11},
  pages={134739--134759},
  year={2023},
  publisher={IEEE}
}

@article{bensalem2024signaling,
  title={Signaling Rate and Performance of {RIS} Reconfiguration and Handover Management in Next Generation Mobile Networks},
  author={Bensalem, Mounir and Jukan, Admela},
  journal={arXiv preprint arXiv:2407.18183},
  year={2024}
}

@inproceedings{adnan2024performance,
  title={Performance Evaluation of {IRS}-Assisted Intra-cell Handover in Vision-Aided mm{W}ave Networks},
  author={Adnan, Alaa and Al-Quraan, Mohammad and Zoha, Ahmed and Imran, Muhammad Ali and Mohjazi, Lina},
  booktitle={2024 IEEE Wireless Communications and Networking Conference (WCNC)},
  pages={1--6},
  year={2024},
  organization={IEEE}
}

@inproceedings{gao2024intelligent,
  title={Intelligent Reflect Surface Assisted {TOPSIS}-based Communication Handover Strategy},
  author={Gao, YaWen and Li, Bo},
  booktitle={2024 7th International Conference on Electronics Technology (ICET)},
  pages={790--795},
  year={2024},
  organization={IEEE}
}

@article{bassoli2021we,
  title={Why do we need {6G}?},
  author={Bassoli, Riccardo and Fitzek, FH and Strinati, E Calvanese},
  journal={ITU Journal on Future and Evolving Technologies},
  volume={2},
  number={6},
  pages={1--31},
  year={2021},
  publisher={Wireless communication systems in beyond 5G era}
}

@article{renzo2019smart,
  title={Smart radio environments empowered by reconfigurable {AI} meta-surfaces: An idea whose time has come},
  author={Renzo, Marco Di and Debbah, Merouane and Phan-Huy, Dinh-Thuy and Zappone, Alessio and Alouini, Mohamed-Slim and Yuen, Chau and Sciancalepore, Vincenzo and Alexandropoulos, George C and Hoydis, Jakob and Gacanin, Haris and others},
  journal={{EURASIP} Journal on Wireless Communications and Networking},
  volume={2019},
  number={1},
  pages={1--20},
  year={2019},
  publisher={Springer}
}

@article{wu2019towards,
  title={Towards smart and reconfigurable environment: Intelligent reflecting surface aided wireless network},
  author={Wu, Qingqing and Zhang, Rui},
  journal={IEEE communications magazine},
  volume={58},
  number={1},
  pages={106--112},
  year={2019},
  publisher={IEEE}
}

@article{yu2021smart,
  title={Smart and reconfigurable wireless communications: From {IRS} modeling to algorithm design},
  author={Yu, Xianghao and Jamali, Vahid and Xu, Dongfang and Ng, Derrick Wing Kwan and Schober, Robert},
  journal={IEEE Wireless Communications},
  volume={28},
  number={6},
  pages={118--125},
  year={2021},
  publisher={IEEE}
}

@article{liao2023optimized,
  title={Optimized design for {IRS}-assisted integrated sensing and communication systems in clutter environments},
  author={Liao, Chikun and Wang, Feng and Lau, Vincent KN},
  journal={IEEE Transactions on Communications},
  volume={71},
  number={8},
  pages={4721--4734},
  year={2023},
  publisher={IEEE}
}

@article{tashan2024optimal,
  title={Optimal handover optimization in future mobile heterogeneous network using integrated weighted and fuzzy logic models},
  author={Tashan, Waheeb and Shayea, Ibraheem and Aldirmaz-Colak, Sultan and El-Saleh, Ayman A and Arslan, H{\"u}seyin},
  journal={IEEE Access},
  year={2024},
  publisher={IEEE}
}

@article{cui2022near,
  title={Near-field rainbow: Wideband beam training for {XL-MIMO}},
  author={Cui, Mingyao and Dai, Linglong and Wang, Zhaocheng and Zhou, Shidong and Ge, Ning},
  journal={IEEE Transactions on Wireless Communications},
  volume={22},
  number={6},
  pages={3899--3912},
  year={2022},
  publisher={IEEE}
}

@article{zhi2024performance,
  title={Performance analysis and low-complexity design for {XL-MIMO} with near-field spatial non-stationarities},
  author={Zhi, Kangda and Pan, Cunhua and Ren, Hong and Chai, Kok Keong and Wang, Cheng-Xiang and Schober, Robert and You, Xiaohu},
  journal={IEEE Journal on Selected Areas in Communications},
  year={2024},
  publisher={IEEE}
}

@article{shen2023multi,
  title={Multi-beam design for near-field extremely large-scale {RIS}-aided wireless communications},
  author={Shen, Decai and Dai, Linglong and Su, Xin and Suo, Shiqiang},
  journal={IEEE Transactions on Green Communications and Networking},
  volume={7},
  number={3},
  pages={1542--1553},
  year={2023},
  publisher={IEEE}
}

@article{liu2022deep,
  title={Deep learning based beam training for extremely large-scale massive {MIMO} in near-field domain},
  author={Liu, Wang and Ren, Hong and Pan, Cunhua and Wang, Jiangzhou},
  journal={IEEE Communications Letters},
  volume={27},
  number={1},
  pages={170--174},
  year={2022},
  publisher={IEEE}
}

@article{peng2024deep,
  title={Deep Learning-Based CSI Feedback for {XL-MIMO} Systems in the Near-Field Domain},
  author={Peng, Zhangjie and Liu, Ruijing and Li, Zhaotian and Pan, Cunhua and Wang, Jiangzhou},
  journal={arXiv preprint arXiv:2405.09053},
  year={2024}
}

@article{han2020channel,
  title={Channel estimation for extremely large-scale massive {MIMO} systems},
  author={Han, Yu and Jin, Shi and Wen, Chao-Kai and Ma, Xiaoli},
  journal={IEEE Wireless Communications Letters},
  volume={9},
  number={5},
  pages={633--637},
  year={2020},
  publisher={IEEE}
}

@article{cui2022channel,
  title={Channel estimation for extremely large-scale {MIMO}: Far-field or near-field?},
  author={Cui, Mingyao and Dai, Linglong},
  journal={IEEE Transactions on Communications},
  volume={70},
  number={4},
  pages={2663--2677},
  year={2022},
  publisher={IEEE}
}

@INBOOK{9966212,
  author={Kazim, Jalil R. and Rains, James and Imran, Muhammad Ali and Abbasi, Qammer H.},
  booktitle={Intelligent Reconfigurable Surfaces ({IRS}) for Prospective {6G} Wireless Networks}, 
  title={Application and Future Direction of {RIS}}, 
  year={2023},
  volume={},
  number={},
  pages={171-188},
  doi={10.1002/9781119875284.ch9}
}

@ARTICLE{8926369,
  author={Tang, Fengxiao and Kawamoto, Yuichi and Kato, Nei and Liu, Jiajia},
  journal={Proceedings of the IEEE}, 
  title={Future Intelligent and Secure Vehicular Network Toward {6G}: Machine-Learning Approaches}, 
  year={2020},
  volume={108},
  number={2},
  pages={292-307},
  doi={10.1109/JPROC.2019.2954595}
}

@ARTICLE{10041749,
  author={Loscrí, Valeria and Rizza, Carola and Benslimane, Abderrahim and Vegni, Anna Maria and Innocenti, Eros and Giuliano, Romeo},
  journal={IEEE Transactions on Vehicular Technology}, 
  title={{BEST-RIM}: A mmWave Beam Steering Approach Based on Computer Vision-Enhanced Reconfigurable Intelligent Metasurfaces}, 
  year={2023},
  volume={72},
  number={6},
  pages={7613-7626},
  doi={10.1109/TVT.2023.3243358}
}

@ARTICLE{10380573,
  author={Alkaabi, Shaimaa R. and Gregory, Mark A. and Li, Shuo},
  journal={IEEE Access}, 
  title={Multi-Access Edge Computing Handover Strategies, Management, and Challenges: A Review}, 
  year={2024},
  volume={12},
  number={},
  pages={4660-4673},
  doi={10.1109/ACCESS.2024.3349587}
}

@ARTICLE{10471522,
  author={Zhang, Bibo and Filippini, Ilario and Lian, Zhuxian and Su, Yinjie},
  journal={IEEE Communications Letters}, 
  title={Adaptive Obstacle-Aware {RIS} Switch for Mobile mm{W}ave Access Networks}, 
  year={2024},
  volume={28},
  number={5},
  pages={1216-1220},
  doi={10.1109/LCOMM.2024.3376713}}

@ARTICLE{8936989,
  author={Özdogan, Özgecan and Björnson, Emil and Larsson, Erik G.},
  journal={IEEE Wireless Communications Letters}, 
  title={Intelligent Reflecting Surfaces: Physics, Propagation, and Pathloss Modeling}, 
  year={2020},
  volume={9},
  number={5},
  pages={581-585},
  doi={10.1109/LWC.2019.2960779}}

@ARTICLE{9003219,
  author={Banagar, Morteza and Chetlur, Vishnu Vardhan and Dhillon, Harpreet S.},
  journal={IEEE Wireless Communications Letters}, 
  title={Handover Probability in Drone Cellular Networks}, 
  year={2020},
  volume={9},
  number={7},
  pages={933-937},
  doi={10.1109/LWC.2020.2974474}}

@ARTICLE{8839948,
  author={Joshi, Kishor Chandra and Hersyandika, Rizqi and Prasad, R. Venkatesha},
  journal={IEEE Systems Journal}, 
  title={Association, Blockage, and Handoffs in {IEEE} 802.11ad-Based 60-{GH}z Picocells—{A} Closer Look}, 
  year={2020},
  volume={14},
  number={2},
  pages={2144-2153},
  doi={10.1109/JSYST.2019.2937568}}

@INPROCEEDINGS{9297352,
  author={Okaf, Abdanaser and Qiu, Dongyu},
  booktitle={2020 International Symposium on Networks, Computers and Communications (ISNCC)}, 
  title={Analysis of Blockage Impact on Handover Rate for User with Mobility in {5G} mm-Wave Cellular Network}, 
  year={2020},
  volume={},
  number={},
  pages={1-6},
  doi={10.1109/ISNCC49221.2020.9297352}}

@INPROCEEDINGS{9615850,
  author={Okaf, Abdanaser and Saied, Amamer and Qiu, Dongyu},
  booktitle={2021 International Symposium on Networks, Computers and Communications (ISNCC)}, 
  title={Analysis of Self-Blockage Impact on Handover Probability for User with Mobility in {5G} mm{W}ave Cellular Network}, 
  year={2021},
  volume={},
  number={},
  pages={1-5},
  doi={10.1109/ISNCC52172.2021.9615850}}

@ARTICLE{8643739,
  author={Jain, Ish Kumar and Kumar, Rajeev and Panwar, Shivendra S.},
  journal={IEEE Journal on Selected Areas in Communications}, 
  title={The Impact of Mobile Blockers on Millimeter Wave Cellular Systems}, 
  year={2019},
  volume={37},
  number={4},
  pages={854-868},
  doi={10.1109/JSAC.2019.2898756}}

@INPROCEEDINGS{10558216,
  author={Palitharathna, Kapila W. S. and Vegni, Anna Maria and Diamantoulakis, Panagiotis D. and Suraweera, Himal A. and Krikidis, Ioannis},
  booktitle={2024 IEEE International Symposium on Circuits and Systems (ISCAS)}, 
  title={Handover Management through Reconfigurable Intelligent Surfaces for {VLC} under Blockage Conditions}, 
  year={2024},
  volume={},
  number={},
  pages={1-5},
  doi={10.1109/ISCAS58744.2024.10558216}}

@article{you2022deploy,
  title={How to deploy intelligent reflecting surfaces in wireless network: {BS}-side, user-side, or both sides?},
  author={You, Changsheng and Zheng, Beixiong and Mei, Weidong and Zhang, Rui},
  journal={Journal of Communications and Information Networks},
  volume={7},
  number={1},
  pages={1--10},
  year={2022},
  publisher={PTP}
}

@article{wei2023equivalent,
  title={An equivalent model for handover probability analysis of {IRS}-aided networks},
  author={Wei, Haoyan and Zhang, Hongtao},
  journal={IEEE Transactions on Vehicular Technology},
  volume={72},
  number={10},
  pages={13770--13774},
  year={2023},
  publisher={IEEE}
}

@article{zhang2024discrete,
  title={Discrete-Time Modeling and Handover Analysis of Intelligent Reflecting Surface-Assisted Networks},
  author={Zhang, Hongtao and Wei, Haoyan},
  journal={arXiv preprint arXiv:2403.07323},
  year={2024}
}

@inproceedings{okaf2021analysis,
  title={Analysis of self-blockage impact on handover probability for user with mobility in {5G} mm-{W}ave cellular network},
  author={Okaf, Abdanaser and Saied, Amamer and Qiu, Dongyu},
  booktitle={2021 International Symposium on Networks, Computers and Communications (ISNCC)},
  pages={1--5},
  year={2021},
  organization={IEEE}
}

@ARTICLE{10078238,
  author={Iqbal, Subhyal Bin and Nadaf, Salman and Awada, Ahmad and Karabulut, Umur and Schulz, Philipp and Fettweis, Gerhard P.},
  journal={IEEE Access}, 
  title={On the Analysis and Optimization of Fast Conditional Handover With Hand Blockage for Mobility}, 
  year={2023},
  volume={11},
  number={},
  pages={30040-30056},
  doi={10.1109/ACCESS.2023.3260630}
}

@inproceedings{mollel2019handover,
  title={Handover management in dense networks with coverage prediction from sparse networks},
  author={Mollel, Michael and Ozturk, Metin and Kisangiri, Michael and Kaijage, Shubi and Onireti, Oluwakayode and Imran, Muhammad Ali and Abbasi, Qammer H},
  booktitle={2019 IEEE Wireless Communications and Networking Conference Workshop (WCNCW)},
  pages={1--6},
  year={2019},
  organization={IEEE}
}

@article{ghosh2020analyzing,
  title={Analyzing handover performances of mobility management protocols in ultra-dense networks},
  author={Ghosh, Shankar K and Ghosh, Sasthi C},
  journal={Journal of Network and Systems Management},
  volume={28},
  number={4},
  pages={1427--1452},
  year={2020},
  publisher={Springer}
}
\end{document}